\documentclass[runningheads]{llncs}

\usepackage[colorlinks]{hyperref}
\usepackage{url}
\usepackage{amsmath}
\usepackage{algorithm}
\usepackage{algpseudocode}
\usepackage{lipsum}
\usepackage{multirow}
\usepackage{booktabs}
\usepackage{graphicx}
\usepackage[normalem]{ulem}
\usepackage{amsmath}
\usepackage{amssymb}
\usepackage{mathrsfs}
\usepackage{float}
\usepackage[colorinlistoftodos]{todonotes}
\newcommand{\NOT}[1]{\ensuremath{{\sim}\,#1}}
\usepackage[misc]{ifsym}

\newcommand{\name}{DAEMON}

\begin{document}

\title{Recommending Related Products Using Graph Neural Networks in Directed Graphs}
\titlerunning{\name}
%

\author{Srinivas Virinchi\inst{1}\Letter  \and
Anoop Saladi\inst{1}\and
Abhirup Mondal\inst{1}}

\toctitle{DAEMON}
\tocauthor{Srinivas, Anoop, Abhirup}
\authorrunning{Srinivas et al.}
%
\institute{International Machine Learning, Amazon, Bengaluru, India\\
\email{\{virins,saladias,mabhirup\}@amazon.com}}
%

%

\maketitle              
\begin{abstract}
Related product recommendation (RPR) is pivotal to the success of any e-commerce service. In this paper, we deal with the problem of recommending related products i.e., given a query product, we would like to suggest top-$k$ products that have high likelihood to be bought together with it. 
Our problem implicitly assumes asymmetry i.e., for a phone, we would like to recommend a suitable phone case, but for a phone case, it may not be apt to recommend a phone because customers typically would purchase a phone case only while owning a phone.
We also do not limit ourselves to complementary or substitute product recommendation. For example, for a specific night wear t-shirt, we can suggest similar t-shirts as well as track pants. So, the notion of relatedness is subjective to the query product and dependent on customer preferences. 
Further, various factors such as product price, availability lead to presence of selection bias in the historical purchase data, that needs to be controlled for while training related product recommendations model.
These challenges are orthogonal to each other deeming our problem non-trivial. 
To address these, we propose \name, a novel Graph Neural Network (GNN) based framework for related product recommendation, wherein the problem is formulated as a node recommendation task on a directed product graph. 
In order to capture product asymmetry, we employ an asymmetric loss function and learn dual embeddings for each product, by appropriately aggregating features from its neighborhood. \name~leverages multi-modal data sources such as catalog metadata, browse behavioral logs to mitigate selection bias and generate recommendations for cold-start products.
Extensive offline experiments show that \name~outperforms state-of-the-art baselines by 30-160\% in terms of HitRate and MRR for the node recommendation task. In the case of link prediction task, \name~presents 4-16\% AUC gains over state-of-the-art baselines. 
\name~delivers significant improvement in revenue and sales as measured through an A/B experiment.
\end{abstract}
\section{Introduction}
\label{intro}
Related product\footnote{We use product, item and node interchangeably in this paper.} recommendation (RPR) plays a vital role in helping customers easily find right products on e-commerce websites and hence, critical to their success. It not only helps customers discover new related products, but also simplifies their shopping effort, thereby delivering a great shopping experience.
In this paper, we are interested in related product recommendation problem: \textit{given a query product, 
the goal is to recommend top-$k$ products that have a high likelihood to be bought together with it.} For notation purpose, let $a$, $b$ and $c$ be any three products. Let $R_{cp}$, $R_{cv}$ and $R_{like}$ be binary relationships where, $aR_{cp}b$ represents that $a$ is co-purchased with $b$, $aR_{cv}b$ represents that $a$ is co-viewed with $b$, and $aR_{like}b$ represents that $a$ is similar to $b$ based on its product features.
\begin{figure}[!htb]%
    \centering
    {\includegraphics[scale=.25]{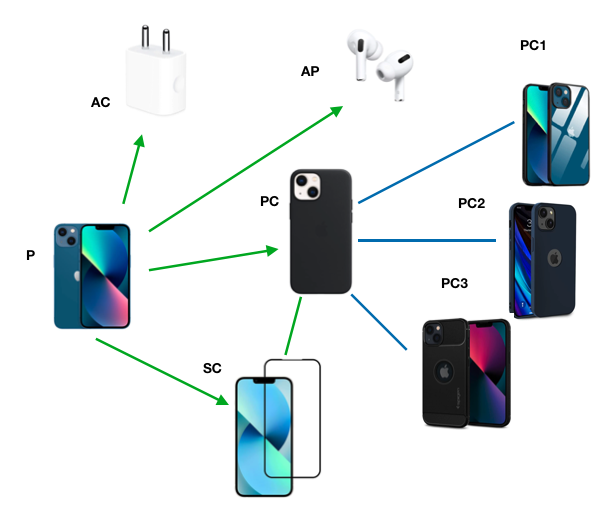} }%
    \caption{Sample product graph. Green and blue edges represent co-purchase and co-view edges respectively. $P$ refers to an iPhone 13, $AC$ is an adapater, $AP$ refers to AirPods and $SC$ refers to a screen guard. Phone case $PC$ is co-viewed with other phone cases $PC1$, $PC2$ and $PC3$; $PC1$, $PC2$ and $PC3$ are similar to PC. Undirected edges are bidirectional.}%
    \label{fig:prob}%
\end{figure}
Consider a sample product graph in Figure~\ref{fig:prob}, where a node corresponds to a product and an edge corresponds to a co-purchase or co-view relationship between the products (refer to Section~\ref{graph-construction}  for graph construction). RPR problem entails few challenges which we illustrate using Figure~\ref{fig:prob} as follows: 1) \textit{product relevance:} Given a query product $P$, we would like to suggest related products such as an adapter $AC$, phone case $PC$, AirPods $AP$, screen guard $SC$. 2) \textit{product asymmetry:} For a phone P, we would like to recommend a suitable phone case PC, but for a phone case PC, it may not be apt to recommend a phone $P$ because customers typically would purchase a phone case only while owning a phone. Formally, product asymmetry can be represented as $a R_{cp}b \not\Rightarrow b R_{cp}a$. 3) \textit{selection bias} is inherent to historical purchase data due to several factors like product availability, price etc. For example, while shopping for a phone case for phone $P$, customer $c_1$ might browse phone  cases $PC$ and $PC1$ ($PC2$ and $PC3$ are not shown to him owing to some of the factors mentioned above). Another customer $c_2$ might browse phone cases $PC$ and $PC2$ ($PC1$ and $PC3$ are not presented to him). A third customer $c_3$ might browse phone cases $PC$ and $PC3$ ($PC1$ and $PC2$ are not presented to him). Assume that all three customers eventually purchase $PC$; this could also be due to a bias in the list of products shown to each customer. In our example, across different customers, $PC$, $PC1$, $PC2$ and $PC3$ are co-viewed (similar). 
In order to correct for selection bias, we would like to recommend not only phone case $PC$, but also $PC1$, $PC2$ and $PC3$ as related products given a query product $P$. 
Formally, we want to uncover relationships of the form: a) $aR_{cp}b \;\land \; bR_{cv}c \implies aR_{cp}c$, b) $aR_{cv}b \;\land \; bR_{cp}c \implies aR_{cp}c$ in order to mitigate selection bias. 4) \textit{cold-start product recommendations:} We not only need to suggest recommendations for existing products, but also suggest recommendations for newly launched i.e. cold-start products. Formally, given two existing products $a$ and $b$, and a cold-start product $c$, we want to uncover relationships of the form: a) $cR_{like}a \;\land \; aR_{cp}b \implies cR_{cp}b$, b) $aR_{cp}b \;\land \; bR_{like}c \implies aR_{cp}c$. These two rules correspond to the case when $c$ is the query product and the recommended product respectively. Observe that mitigating selection bias and tackling cold-start products deal with modelling transitive relationships, while preserving edge asymmetry. 5) Dealing with millions of products in our catalog demands a \textit{scalable} solution for the recommendation task. These challenges are uncorrelated with each other making our problem non-trivial.

Given multi-modal data sources such as catalog metadata, product co-purchase data, anonymized browse behavioral logs, etc. product graphs serve as an excellent abstraction to seamlessly capture relationships between products.
RPR boils down to the node recommendation problem~\cite{app,hope} in directed product graphs. Recent work on node representation learning in directed graphs such as HOPE~\cite{hope}, APP~\cite{app}, NERD~\cite{nerd}, DGGAN~\cite{zhu2020adversarial} and Gravity Graph VAE~\cite{salha2019gravity} learn two embeddings for each node and utilize them for downstream tasks. Further, GNN models like DGCN~\cite{dgcn} and DiGraphIB~\cite{digcn} generate real valued embeddings for each node, while MagNet~\cite{magnet} generates complex valued embeddings to preserve edge strength and direction. 
Selection bias in recommendation systems has been well studied~\cite{chen2020bias} for different applications. 
Prior work has limitations across different dimensions as follows: a) there is no straightforward approach to extend popular GNN models like GCN~\cite{gcn}, GraphSage~\cite{sage}, GAT~\cite{gat}, RGCN~\cite{rgcn} to model directed graphs as they do not model product asymmetry. b) GNN models for directed graphs like DGCN~\cite{dgcn}, DGGAN~\cite{zhu2020adversarial}, DiGraphIB~\cite{digcn} and MagNet~\cite{magnet} do not address the problem of recommendation for cold-start products. c) prior work does not address selection bias inherent to historical purchase data, and make specific assumption regarding the availability of unbiased data~\cite{wang2021combating}.


In this paper, we formulate RPR as a node recommendation task on a directed product graph. We propose \name, \textbf{D}irection \textbf{A}war\textbf{E} Graph Neural Network \textbf{MO}del for \textbf{N}ode recommendation. In order to capture product asymmetry, our model generates dual embeddings for each product, by appropriately aggregating features from its neighborhood. We leverage customer co-view pairs from anonymized browse behavioral logs to recommend relevant products that customers clicked but didn't purchase to tackle selection bias. Specifically, during model training, we employ a novel asymmetric loss function which explicitly considers product co-purchases to capture product asymmetry, and product co-views to estimate co-purchase likelihood for products that were previously clicked but not purchased. 
We exploit product catalog metadata to generate product embeddings for cold-start products, and consequently generate recommendations for them. Further, as a byproduct, the proposed model is also able to suggest substitute product recommendations. 

We perform exhaustive experiments offline on real-world datasets to evaluate the performance of our model. Results show that \name~outperforms state-of-the-art baselines by 30-160\% in terms of HitRate and MRR for the node recommendation task which is the primary task of interest. For link prediction tasks, \name~outperforms state-of-the-art baselines by 4-16\% and 3-6.5\% in terms of AUC for the existence~\cite{magnet} and direction~\cite{magnet} link prediction task respectively. \name~delivers significant improvement in revenue and sales as measured through an A/B experiment.
To summarize, we make the following contributions: 
\begin{enumerate}
    \item We formulate RPR as a node recommendation task in a directed product graph and propose a Graph Neural Network (GNN) based framework for related product recommendation. To this end, we present \name, a novel GNN model, that leverages dual embeddings to capture node asymmetry in directed graphs. 
    \item In order to train the model, we employ a novel asymmetric loss function that explicitly deals with modelling co-purchase likelihood, product asymmetry and selection bias. We exploit product catalog metadata to deal with the issue of cold-start products. 
    This is the first work that jointly addresses product asymmetry, generating cold start recommendations and mitigating selection bias in historical purchase data.
    \item Offline evaluation shows that \name~outperforms state-of-the-art models for node recommendation and link prediction tasks. 
    \name~derives significant improvement in product sales and revenue as measured through an A/B experiment. 
\end{enumerate}
 
\section{Related Work}
\label{relatedwork}
Our problem conceptually relates to the node recommendation problem~\cite{app,hope} in directed graphs. Prior work uses random walk based techniques to model node relationships in directed graphs.
VERSE~\cite{verse}, HOPE~\cite{hope} propose learning two embeddings for each node to preserve higher order proximity and consequently, node asymmetry in directed graphs. APP~\cite{app} captures asymmetry by preserving Rooted PageRank between nodes by relying on random walk with restart strategy. ATP~\cite{sun2019atp} addresses the problem of question answering by embeddings nodes of directed graph by preserving asymmetry. However, their approach is restricted to only directed acyclic graphs (DAGs), while an e-commerce product graph can be cyclic. NERD~\cite{nerd} learns a pair of role specific embeddings for each node using a alternating random walk strategy to model edge strength and direction in directed graphs. 

Recent work has seen GNN's being designed for directed graphs. DGCN~\cite{dgcn} extends the spectral-based GCN model to directed graphs by using first and second-order proximity to expand the receptive field of the convolution operation. APPNP~\cite{appnp} uses a GCN model to approximate personalized PageRank. DGCN~\cite{dgcn} uses one first-order matrix and two second order proximity matrices to model asymmetry. 
DiGraphIB~\cite{digcn} builds upon the ideas of \cite{dgcn} and constructs a directed Laplacian of a PageRank matrix. It uses an inception module to share information between receptive fields. Gravity GAE~\cite{salha2019gravity} is inspired from Graph Auto Encoders~\cite{kipf2016variational} by applying the idea of gravity to address link prediction in directed graphs. DGGAN~\cite{zhu2020adversarial} is based on Generative Adversarial Network by using a discriminator and two generators to jointly learn each node's source and target embedding. MagNet~\cite{magnet} proposes a GNN for directed graphs based on a complex Hermitian matrix. The magnitude of entries in the complex matrix encodes the graph structure while the directional aspect is captured in the phase parameter. Prior work in this space does not deal with the the issue of cold-start products. Selection bias in recommendation systems has been well studied~\cite{chen2020bias} for different applications. However, prior work does not address the bias inherent to historical purchase data. Further, we make no assumption regarding the availability of unbiased data in order to combat selection bias unlike \cite{wang2021combating}. Ours is the first work that jointly: a) learns node representations in directed graphs to capture edge strength and direction, b) generates recommendations for cold-start products by leveraging catalog metadata, and c) mitigates selection bias in product co-purchase directed graphs. 

\section{Related Product Recommendation Problem}
\begin{table}[!htb]
  \caption{Notation}
  \label{tab:notation}
  \centering
  \scalebox{0.8}{
  \begin{tabular}{ll}
    \toprule
    \textbf{Notation}     & \textbf{Description}  \\
    \midrule
    $P$ & set of products in catalog \\
    $i \in P$ & product $i$ in catalog \\
    $X_i$ & input feature of product $i$ \\
    $\theta^{s}_i$ & source embedding of product $i$\\
    $\theta^{t}_i$ & target embedding of product $i$\\
    $G$ & directed product graph\\
    $q \in P$ & query product  \\ 
    $R^q_k$ & top-$k$ related products for query product $q$\\
    $CP$ & product co-purchase pairs\\
    $CV$ & product co-view pairs\\
    $E_{cp}$ & product co-purchase edges \\
    $E_{cv}$ & product co-view edges \\
    \bottomrule
  \end{tabular}}
\end{table}
We present relevant notation in Table~\ref{tab:notation}.
Let $G$ be a directed product graph with products $P$ as the nodes and directed edges corresponding to a co-purchase or co-view relationship between the products (refer section~\ref{graph-construction}). Further, every product $i$ ($i \in P, \forall i$) has an input feature $X_i$ from the product catalog metadata. Given a query node\footnote{We use product and node interchangeably in this paper based on the context.} $q$, the goal is to recommend $R^q_k$, top-$k$ related products that have a high likelihood to be bought together with $q$. 

This corresponds to the node recommendation problem~\cite{app,hope} in directed graphs. Note that while generating related product recommendations, we must jointly capture product co-purchase likelihood and preserve product asymmetry. 
\section{Proposed Framework}
\begin{figure}[!htb]%
    \centering
    {\includegraphics[height=8cm, width=11cm]{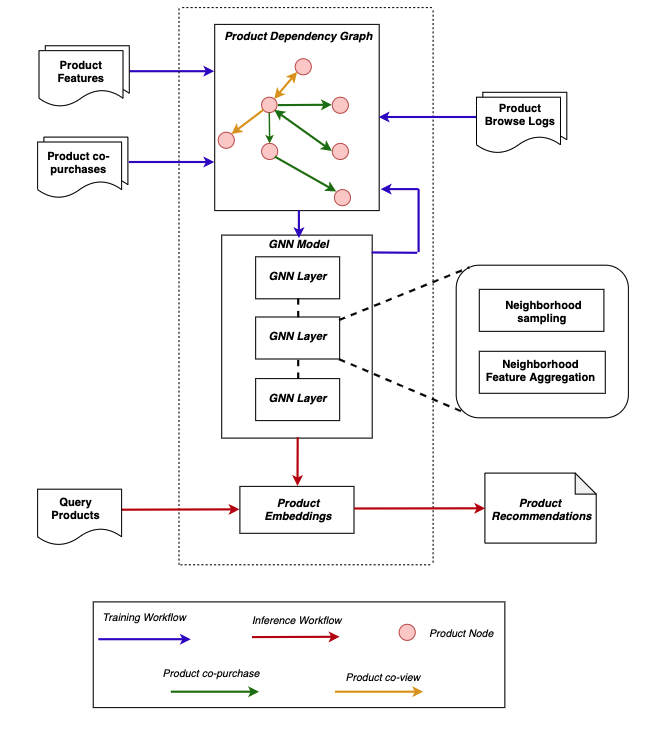} }%
    \caption{Proposed Framework for Related Product Recommendation. The computational graphs corresponding to the source (s) and target (t) embedding of node $A$ in the product graph is shown.}
    \label{fig:framework}%
\end{figure}
We show the proposed framework in Figure~\ref{fig:framework}. We leverage the co-purchase data $CP$ and co-view data $CV$ to create the product graph $G$. We represent each product by an input feature based on its metadata from catalog (product name, product description, product type etc.). In order to jointly model product co-purchase likelihood and product asymmetry, we represent each product using dual i.e. source and target embeddings. We train \name~using an asymmetric loss function.
The trained model generates product embeddings by appropriately aggregating features from its neighborhood in $G$.  We show the computational graphs corresponding to the source and target embedding of node $A$ in Figure~\ref{fig:framework}. For every node $u$, we generate a pair of embeddings (denoted by u-s and u-t in Figure~\ref{fig:framework}). In order to address node asymmetry, we use the following rationale: a) while generating the source embedding of a node, we aggregate information from the target embedding of its out-neighbors, b) while generating the target embedding of a node, we aggregate information from the source embedding of its in-neighbors. Compared to undirected graphs, the computational graph corresponding to the source and target embedding of each node is different in directed graphs. Consequently, this trick enables us to model edge strength and edge direction efficiently in directed graphs. Given a query product $q$, we generate its source embedding and perform a nearest neighbor lookup in the target embedding space to recommend top-$k$ related products for $q$. 
We explain each component in detail next.

\subsection{Product graph construction}
\label{graph-construction}
Given a set of product co-purchase pairs $CP$, we create the co-purchase edges $E_{cp} = \{(u,v) | \forall\;u,v \in P \land uR_{cp}v \}$. We use $E_{cp}$ to model the product co-purchase likelihood. However, $E_{cp}$ is prone to selection bias.  

In order to combat this issue, we exploit anonymized browse behavioral logs, to create product co-view pairs $CV$, which we use to create the set of co-view edges $E_{cv} = \{(u,v) | \forall\;u,v \in P \land uR_{cv}v\}$.
This lets us construct a directed product graph $G = (P, \{E_{cp} \cup E_{cv} \})$, which contains both co-purchase and co-view relationships between products. 

For any three products $a,b,c \in P$, $G$ contains product relationships of the form: 1) $aR_{cp}b \not\Rightarrow bR_{cp}a$ 2) $aR_{cp}b \;\land \; bR_{cv}c \implies aR_{cp}c$ 3) $aR_{cv}b \;\land \; bR_{cp}c \implies aR_{cp}c$. The goal is to design a GNN model that learns these kind of relationships from $G$, while preserving co-purchase likelihood; rule 1 corresponds to product asymmetry, and rules 2-3 are specific to the case of selection bias.
\textit{Observe that $E_{cp}$ and $E_{cv}$ are not correlated with each other, and we need to provide differential treatment to $E_{cp}$ and $E_{cv}$ during neighborhood feature aggregation.} 

\subsection{\name : Proposed GNN model}
\label{train}

We first describe the product embedding generation procedure i.e. forward pass of the model assuming that the model is already trained (Section~\ref{forward}). We then describe how the model parameters can be learned using stochastic gradient descent and backpropagation techniques (Section~\ref{backward}). 

\subsubsection{Product Embedding Generation: Forward Pass}
\label{forward}
 Algorithm~\ref{algo:forward} describes the forward pass to generate the source and target embedding of each product. It expects the product graph $G$ and product features $X$ as input.  
\begin{algorithm}[!htb]
  \caption{\name~ product embedding generation (i.e. forward pass) \label{algo:forward}}
  \textbf{Input: } product graph $G = (P, \{E_{cp} \cup E_{cv} \})$; input product features $\{X_u, \forall \in P\}$; Number of GNN layers L; weight matrices $W^l,\forall l\in \{1,2,..,L\}$\\
  \textbf{Output:} source embedding $\theta^{s}_u$ and target embedding $\theta^{t}_u$, $\forall u \in P$ 
  \begin{algorithmic}[1]

    \State  $({h^s_u})^0 \gets X_u$ ; $({h^t_u})^0 \gets X_u \forall u \in P$
    
    \For{$l$=1,...$L$} 
        \For{$u \in P$} 
            \State $({h^s_u})^l \gets \;\; \sigma\left( \sum_{(u,v) \in E_{cp}} ({h^t_v})^{l-1} W^l \right)$
                                              + $\sigma\left(\sum_{(u,v) \in E_{cv}} {(h^s_v})^{l-1} W^l \right)$
                                              
                                              
            \State $({h^t_u})^l \gets \;\; \sigma\left(\sum_{(v,u) \in E_{cp}} {(h^s_v})^{l-1} W^l \right) + \sigma\left( \sum_{(v,u) \in E_{cv}} {(h^t_v})^{l-1} W^l \right)$
        \EndFor 
        \State $({h^s_u})^l \gets ({h^s_u})^l$/$||({h^s_u})^l||_2, \forall u \in P$
        \State $({h^t_u})^l \gets ({h^t_u})^l$/$||({h^t_u})^l||_2, \forall u \in P$
    \EndFor 
    \State $\theta^{s}_u \gets ({h^s_u})^L,\forall u \in P$
    \State $\theta^{t}_u \gets ({h^t_u})^L,\forall u \in P$
\end{algorithmic}
\end{algorithm}
We set the initial source and target embedding of each product to its input feature. Let $({h^s_u})^l$ and $({h^t_u})^l$ denote a product's source and target representation in the $l^{th}$ step of Algorithm~\ref{algo:forward}. First, for each product $u$, extract its co-purchase and co-view neighbors. The source hidden node representation of $u$ in the $l^{th}$ step is aggregated as a linear combination of two non-linear terms i.e. non-linear aggregate of target representation of its co-purchased out-neighbors and non-linear aggregate of source representation of its co-viewed out-neighbors from the $(l-1)^{th}$ step. $W^l$ corresponds to fully connected layer at step $l$ with non-linear ReLU activation function $\sigma$ (line $4$). We perform a similar neighborhood aggregation to estimate a node's target representation during the $l^{th}$ step (line $5$). The source and target node representations of step $l$ is used in the next step.
We normalize the embeddings to a unit norm (lines $7$, $8$). We repeat this process for $L$ steps to generate the final source and target product representation of products $\{\theta^{s}_u, \theta^{t}_u\}\;\forall u \in P$ (line $10$, $11$). We leverage the generated product embeddings to recommend related products. We also show relevant properties of the embeddings in Lemma~\ref{lemma1} and \ref{lemma2}. \\
\textbf{Related product recommendation:} 
Given a query product $q$, we use $\theta^{s}_q$, the source embedding of $q$, to perform a nearest neighbor lookup in the target embedding space of all the products to recommend top-$k$ set of related products denoted by $R^q_k$. Specifically, for a query product $q \in P$, we compute a relevance score with respect to a candidate product $v \in P$ as shown in Equation~\ref{eq:rel}.
\begin{eqnarray}
\label{eq:rel}
rel(q,v) = (\theta_q^s) ^\intercal (\theta_v^t)
\end{eqnarray}
Observe that $rel(q,v) \neq rel(v,q)$ which helps capturing product asymmetry. \\
\textbf{Cold-start related product recommendation:} Our model can also be used for \textit{cold-start} product recommendation. Given a cold-start product $c$, we perform a neighborhood lookup to find similar existing (warm-start) products i.e. $\{c_1, c_2,..c_k\}$ based on its input product features ($X$). We augment the edges of the form $\{(c, c_1)$, $(c, c_2)$,..$(c, c_k)\}$ to the product graph $G$. Further, $c$ has an input feature $X_c$ from catalog metadata. We pass the subgraph corresponding to the cold-start product $c$ to the GNN model and generate the source and target embedding of $c$ using Algorithm~\ref{algo:forward}. We use $\theta^s_c$ to probe the target embeddings to recommend related products for $c$ using Equation~\ref{eq:rel}. 
\subsubsection{Learning \name~parameters: Backward Pass}
\label{backward}
In order to train the model in an unsupervised manner, 
we employ a novel asymmetric loss function shown in Equation~\ref{eq:gnnloss}. This helps us capture different aspects as follows:
\begin{enumerate}
    \item When $(u,v) \in E_{cp}$, the model forces the source embedding of $u$ to be similar to the target embedding of $v$, and distant from the target embedding of a disparate product $z$ (terms $1$ and $2$). Negative samples are generated using a uniform distribution $P_r$. This models product co-purchase likelihood and asymmetry.
    \item In order to force product asymmetry, for every one-way directed co-purchase edge i.e. $(u,v) \in E_{cp} \land (v,u) \notin E_{cp}$, our model assigns a high score to the edge $(u,v)$ and a low score to the edge $(v,u)$ (terms $3$ and $4$). 
    \item For similar products i.e. $(u,v) \in E_{cv}$, our model forces both the source and target embeddings of $u$ and $v$ to be similar (terms $5$ and $6$). This property is useful in cold-start recommendation and mitigating selection bias.  
\end{enumerate}
\begin{align}
 \label{eq:gnnloss}
    loss &=  - \left\{ \sum_{(u,v) \in E_{cp}} log(\sigma(\theta^s_u\cdot\theta^t_v)) + 
     \sum\limits_{\substack{n_s =1 \\ z \sim\ P_r(P), u \neq z, } }^{n_k} log(\sigma(1-\theta_u^s\cdot\theta_z^t)) \right\} \nonumber\\   
      &- \left \{ \sum\limits_{\substack{(u,v) \in E_{cp} \land \\ (v,u) \notin E_{cp}  } } log(\sigma(\theta_u^s\cdot\theta_v^t))  +
     \sum\limits_{\substack{(u,v) \in E_{cp} \land \\ (v,u) \notin E_{cp}  } } log(\sigma(1 - \theta_v^s\cdot\theta_u^t)) \right\} \nonumber\\ &- \left \{ \sum_{(u,v) \in E_{cv}} log(\sigma(\theta^s_u\cdot\theta^s_v)) + 
    \sum_{(u,v) \in E_{cv} } log(\sigma(\theta^t_u\cdot\theta^t_v)) 
    \right\}  \tag{2} 
\end{align}\\
\begin{lemma}
\label{lemma1}
The embeddings generated by \name~capture product co-purchase likelihood and product asymmetry.
\end{lemma}
\begin{proof}
When $(a,b) \in E_{cp} \leftrightarrow aR_{cp}b$, our model assigns a high score to $\theta^{s}_a.\theta^{t}_b$ compared to $\theta^{s}_a.\theta^{t}_z$ i.e. $\theta^{s}_a.\theta^{t}_b >> \theta^{s}_a.\theta^{t}_z$ (for a random product $z$). This implies that the co-purchase likelihood between related products $a$ and $b$ is preserved. 
Further, when $(a,b) \in E_{cp} \land (b,a) \notin E_{cp} \leftrightarrow aR_{cp}b \land \NOT bR_{cp}a$, our model will assign a higher score to $\theta^{s}_a.\theta^{t}_b$ compared to $\theta^{s}_b.\theta^{t}_a$ i.e. $\theta^{s}_a.\theta^{t}_b >> \theta^{s}_b.\theta^{t}_a$, thereby capturing product asymmetry.
\end{proof}

\begin{lemma}
\label{lemma2}
The embeddings generated by \name~helps mitigating selection bias.
\end{lemma}
\begin{proof}
For $a,b,c\in P$, $G$ consists paths of the form $\{a,b,c\}$ where, $(a,b) \in E_{cp} \leftrightarrow aR_{cp}b$ and $(b,c) \in E_{cv} \leftrightarrow bR_{cv}c$. Our model assigns a high score to $\theta^{s}_a.\theta^{t}_b$ as $(a,b) \in E_{cp}$. As $(b,c) \in E_{cv}$, our model assigns a high score to $\theta^{s}_b.\theta^{s}_c$ and $\theta^{t}_b.\theta^{t}_c$. Hence, our model assigns a high score to $ \theta^{s}_a.\theta^{t}_c = \theta^{s}_a.\theta^{t}_b \times \theta^{t}_b.\theta^{t}_c$, implying $aR_{cp}c$. We can present a similar argument to show that the embeddings capture $aR_{cv}b \;\land \; bR_{cp}c \implies aR_{cp}c$. In this manner, our model mitigates selection bias.
\end{proof}
In terms of scalability, we employ minibatch sampling both during training and inference procedure to generate product embeddings. We use FAISS~\cite{faiss} to perform efficient nearest neighbor lookup. Our model uses O($|P|$) space to store embeddings corresponding to $|P|$ products. 
We discuss the results in the next section.

\section{Experiments}
As previously discussed, RPR problem is equivalent to the node recommendation task in directed graphs. In this section, we evaluate the performance of \name~against state-of-the-art models in directed graphs. Specifically, we aim to answer the following evaluation questions:
\begin{itemize}
    \item \textbf{EQ1:} Is \name~ able to improve the node recommendation performance on real-world e-commerce datasets?
    \item \textbf{EQ2:} How effective is \name~in capturing the co-purchase likelihood between products?
    \item \textbf{EQ3:} Is \name~able to capture product asymmetry?
    \item \textbf{EQ4:} How effective is \name~for cold-start product recommendation?
    \item \textbf{EQ5:} Is \name~able to combat selection bias?
    \item \textbf{EQ6:} Semantically, what kind of relationships between products can be extracted using the source and target embeddings generated by \name?
\end{itemize}
To this end, we introduce the datasets, baselines and follow it up with experiments to answer these questions. 

\subsection{Experimental Setting}

\subsubsection{Datasets} We extract two real-world datasets sampled from different emerging marketplaces in Amazon. For each dataset, we construct a graph consisting of products as nodes and edges corresponding to either a co-purchase or co-view relationship as explained in Section~\ref{graph-construction}. Note that we use anonymized browse logs while creating product co-view pairs. The statistics of the graph datasets is shown in Table~\ref{datasets}. These graphs consist $2$-$5.5$M nodes and $14$-$32$M edges and are directed ($\NOT 75\%$ of the edges are directed). In both the datasets, we represent each product using input features of size $384$ and $512$ respectively. These datasets are sampled, and not reflective of production traffic in terms of scale.
\begin{table}[!htb]
\caption{Summary of graph datasets employed}
\label{datasets}
\centering
\scalebox{0.7}{
\begin{tabular}{cccccccc} 
\toprule
\multicolumn{1}{l}{\textbf{Dataset}} & \multicolumn{1}{l}{\textbf{\# Nodes}} & \multicolumn{1}{l}{\textbf{\# Edges}} & \begin{tabular}[c]{@{}c@{}}\textbf{Average}\\\textbf{ Degree}\end{tabular} & \begin{tabular}[c]{@{}c@{}}\textbf{\%Directed}\\\textbf{ Edges}\end{tabular} & \begin{tabular}[c]{@{}c@{}}\textbf{\# Co-purchase }\\\textbf{ pairs }\end{tabular} & \begin{tabular}[c]{@{}c@{}}\textbf{\# Co-view}\\\textbf{ pairs}\end{tabular}  & \begin{tabular}[c]{@{}c@{}}\textbf{Input Product}\\\textbf{ Feature Dimension}\end{tabular}\\ 
\midrule
G1                                   & 1.98M                            & 14.1M                            & 7.3                                                                        & 76.33                                                                        & 7M                                                                            & 7.1M      & 384                                                               \\
G2                                   & 5.5M                             & 31.7M                            & 5.76                                                                       & 79.97                                                                        & 13.2M                                                                         & 18.5M        & 512                                                            \\
\bottomrule
\end{tabular}}
\end{table}
\subsubsection{Implementation Details}
We implemented \name~using DGL and PyTorch. We vary the learning rate in $\{ 10^{-1}, 10^{-2}, 10^{-3}, 10^{-4}\}$ and observed $10^{-4}$ using Adam's optimizer to work the best. We use $L=3$ layer GNN model for \name. We use minibatches of size $1024$ during model training (Section~\ref{backward}) and inference (Section~\ref{forward}). We also use layer-wise neighborhood sampling, i.e. $20$ neighbors for the first layer and $10$ random neighbors for the subsequent layers. After generating the product embeddings, we perform nearest neighbor lookup to suggest top-$k$ related products for $k$ values in $\{5,10,20\}$. All the experiments were conducted on a 64-core machine with a 488 GB RAM running Linux. For all models, we learn $64$ dimensional product embeddings trained for a maximum of $30$ epochs. We repeat all the experiments $10$ times and report the average value across the runs. 

\subsubsection{Baselines}
We previously discussed state-of-the-art models in directed graphs (Section~\ref{relatedwork}). In order to evaluate our proposed model, \name, we choose the most competitive baselines as follows:
\begin{enumerate}
    \item We choose APP~\cite{app} and NERD~\cite{nerd} as they deliver superior performance compared to deepwalk\cite{deepwalk}, node2vec~\cite{node2vec}, LINE~\cite{line}, HOPE\cite{hope} and VERSE~\cite{verse}.
    \item Gravity GAE~\cite{salha2019gravity} and DGGAN~\cite{zhu2020adversarial} outperform deepwalk~\cite{deepwalk}, node2vec~\cite{node2vec}, LINE~\cite{line}, APP\cite{app}, HOPE~\cite{hope} and VGAE~\cite{kipf2016variational}.
    \item MagNet~\cite{magnet} is the latest state-of-the-art GNN model which delivers best results compared to GraphSage~\cite{sage}, GAT~\cite{gat}, DGCN~\cite{dgcn} and APPNP~\cite{appnp}.
\end{enumerate}
In summary, we choose APP~\cite{app}, NERD~\cite{nerd}, Gravity GAE~\cite{salha2019gravity}, DGGAN~\cite{zhu2020adversarial} and MagNet~\cite{magnet} as the most competitive baselines. We use publicly released code repositories for all the employed baselines. Further, we employ parameter tuning to choose the best parameters for each baseline. 

However, all the baselines are suitable only for homogeneous directed graphs. For fair comparison, we restrict the evaluation to co-purchase directed graphs (edges corresponding to $E_{cp}$), and evaluate \name~against the baselines (Section~\ref{task1}, Section~\ref{task2}). For the case of complete data ($E_{cp} \cup E_{cv}$), we compare \name~against state-of-the-art R-GCN~\cite{rgcn} model in Section~\ref{task3}.
\subsubsection{Experiment Setup} 
In order to answer the evaluation questions, 
we setup the experiments (similar to ~\cite{app}, ~\cite{magnet}, ~\cite{sage}) as follows:
\begin{itemize}
    \item $[$EQ1$]$. Node Recommendation Task (Section~\ref{task1}): For each graph dataset, we use $75\%$, $5\%$ and $20\%$ non-overlapping edges for training, validation and testing respectively. For a ground-truth recommendation $(u,v)$ (from the test data), where $u$ is the query node, we retrieve top-$k$ node recommendations $(R^q_k)$ suggested by each model. In order to evaluate the quality of recommendations, we use HitRate@k and MRR@k for k in $\{5,10,20\}$. This helps answering EQ1. 
    \item $[$EQ2$]$. Existential link prediction Task (Section~\ref{task2}): For each graph dataset, we use $75\%$, $5\%$ and $20\%$ non-overlapping edges for training, validation and testing respectively. In this task, we want to capture the edge score predicted by each model; a good model assigns higher scores to existing links compared to non-existing links. We evaluate the performance of various models using AUC for this task. This helps answering EQ2.
    \item $[$EQ3$]$. Directed link prediction Task (Section~\ref{task2}): For each graph dataset, we use $75\%$, $5\%$ and $20\%$ non-overlapping edges for training, validation and testing respectively. In this task, we consider only one-way directed edges ($(u,v) \in G \land (v,u) \notin G$) as the test edges. We reverse the direction of the test edges to create negative links for testing. We want to evaluate how different models capture the edge direction; a good model assigns a high score to the correct edge (positive links in the test set) and a low score to the reverse edge (negative links in the test set). We evaluate the model performance using AUC for this task. This helps answering EQ3.
    \item $[$EQ4$]$. Cold-start Recommendation Task (Section~\ref{task3}): For each graph dataset, we use subgraphs corresponding $75\%$, $5\%$ and $20\%$ non-overlapping nodes for training, validation and testing respectively. For evaluation, we take the same approach as employed in the node recommendation task. This helps answering EQ4.
    \item $[$EQ5$]$. Selection-bias Recommendation Task (Section~\ref{task3}): For each graph dataset, we use edges corresponding $75\%$, $5\%$ and $20\%$ non-overlapping nodes for training, validation and testing respectively. \textit{Further, we add edges corresponding to transitive relationships $aR_{cp}b \land bR_{cv}c \implies aR_{cp}c$ to the test set.} In order to mitigate selection bias, a model needs to recommend $c$ as a related product for $a$ as these transitive relationships are present in the training graph. For evaluation, we take the same approach as employed in the node recommendation task. This helps answering EQ5.
\end{itemize}
\subsection{EQ1. Node Recommendation Task on Co-purchase Data}
\label{task1}
The results\footnote{Results are relative to the co-purchase baseline and absolute numbers are not presented due to confidentiality.} for the node recommendation task is shown in Table~\ref{tab:eval1}. We see that \name~performs the best for this task in terms of both HitRate and MRR. 
\name~yields more gains on the bigger $G2$ dataset compared to the $G1$ datset.
Further, DGGAN
fails to complete model training in 48 hours. It does not scale to real-world datasets; the paper show their model efficacy on smaller datasets (the biggest graph has ~15K nodes). Further, Gravity GAE runs out of memory (488 GB RAM) during model training. This is due to one-hot feature encoding employed by the model to represent node input features i.e. each node is represented as a million valued embedding initially. APP delivers the second best performance. MagNet is not applicable for node recommendation task, and we compare against it for the link prediction tasks.
\begin{table}[!htb]
\centering
\caption{Node recommendation. Best results are in \textbf{bold} and second are \uline{underlined}}
\label{tab:eval1}
\scalebox{0.8}{
\begin{tabular}{cccccccc} 
\toprule
\multirow{2}{*}{\begin{tabular}[c]{@{}c@{}}\\Dataset\\~\\\end{tabular}} & \multirow{2}{*}{Model} & \multicolumn{3}{c}{HitRate@k}                       & \multicolumn{3}{c}{MRR@k}                            \\ 
\cmidrule{3-8}
                                                                        &                        & 5               & 10              & 20              & 5               & 10              & 20               \\ 
\midrule
\multirow{5}{*}{G1}                                                     & APP                    & \uline{2.14x} & \uline{4.01x } & \uline{6.95x } & \uline{1.02x } & \uline{1.26x } & \uline{1.46x }  \\
                                                                        & NERD                   & 0.62x          & 1.94x          & 0.0314          & 0.61x          & 0.71x          & 0.79x           \\
                                                                        
                                                                        & \name           & \textbf{3.07x} & \textbf{5.54x} & \textbf{8.98x} & \textbf{1.45x} & \textbf{1.77x} & \textbf{2.01x}  \\ 
\cmidrule{2-8}
                                                                        & \%Gain                 & 43.4            & 38.15           & 29.2            & 42.15           & 40.47           & 37.6             \\ 
\midrule
\multirow{5}{*}{G2}                                                     & APP                    & 1.23x          & \uline{1.99x } & \uline{3.31x } & 0.67x          & 0.77x          & \uline{0.85x }  \\
                                                                        & NERD                   & \uline{1.28x } & 1.85x          & 2.63x          & \uline{0.71x } & \uline{0.78x } & 0.84x           \\
                                                                        & \name           & \textbf{3.51x} & \textbf{5.87x} & \textbf{8.6x}  & \textbf{1.71x} & \textbf{2.02x} & \textbf{2.21x}  \\ 
\cmidrule{2-8}
                                                                        & \%Gain                 & 174             & 194             & 159.8           & 155.2           & 162.3           & 160              \\
\bottomrule
\end{tabular}}
\end{table}

\subsection{[EQ2,EQ3.] Link Prediction Tasks on Co-purchase Data}
\label{task2}

\begin{table}[H]
\centering
\caption{Link prediction (\%). Best results are in \textbf{bold} and second are \uline{underlined}}
\label{tab:eval2}
\scalebox{0.8}{
\begin{tabular}{cccc} 
\toprule
Dataset & Model        & \begin{tabular}[c]{@{}c@{}}Existence Prediction\\(AUC \%)\end{tabular} & \begin{tabular}[c]{@{}c@{}}Direction Prediction~\\(AUC \%)\end{tabular}  \\ 
\midrule
\multirow{7}{*}{G1}                                     & APP          & 34.9x                                                                   & 0.1x                                                                     \\
                                                        & NERD         & 23.3x                                                                  & 7.1x                                                                    \\
                                                        & MagNet       & \uline{36.1x }                                                          & \uline{12.65x}                                                           \\
                                                        & \name & \textbf{40.31x}                                                         & \textbf{14.72x}                                                           \\ 
\cmidrule{2-4}
                                                        & $\Delta$ Gain         & 4.2                                                                    & 2.25                                                                     \\ 
\midrule
\multirow{7}{*}{G2}                                     & APP          & 5.27x                                                                  & 0.15x                                                                   \\
                                                        & NERD         & 14.45x                                                                  & 2.48x                                                                   \\
                                                        & MagNet       & \uline{18.26x }                                                         & \uline{6.62x}                                                           \\
                                                        & \name & \textbf{34.28x}                                                         & \textbf{12.53x}                                                           \\ 
\cmidrule{2-4}
                                                        & $\Delta$ Gain         & 16                                                                     & 6.6                                                                      \\
\bottomrule
\end{tabular}}
\end{table}

We present the performance\footnote{Results are relative to the co-purchase baseline and absolute numbers are not presented due to confidentiality.} of different models in Table~\ref{tab:eval2} for two link prediction tasks. We see that \name~outperforms all baselines on both these tasks. MagNet displays the second best performance. We make a note that predicting existence of links is a simpler task when compared to predicting the correct edge direction; this is also observed in the results. 
As explained previously, DGGAN fails to complete model training in 48 hours and Gravity GAE runs out of memory (488GB RAM) during model training. These results indicate that the proposed model is able to jointly model co-purchase likelihood and product asymmetry compared to the baselines which answers \textbf{EQ2} and \textbf{EQ3}.

\subsection{[EQ4, EQ5, EQ6.] Ablation Study on G1 graph}
\label{task3}
We evaluate\footnote{Results are relative to the R-GCN baseline and absolute numbers are not presented due to confidentiality.} \name~on the complete graph in this section to analyze its performance for the case of cold-start product recommendation task (Table~\ref{tab:eval3}) and the case of mitigating selection bias (Table~\ref{tab:eval32}). Observe from Table~\ref{tab:eval3} that \name~delivers significant performance gains over R-GCN for cold-start recommendation evaluation. Although, both models were fed with the same cold-start graphs, observe that R-GCN performs poorly. Results indicate how leveraging catalog product information to create subgraphs pertaining to cold-start products aid in boosting the performance of our model. This answers \textbf{EQ4}.

\begin{table}
\centering
\caption{Cold-start evaluation on G1 graph}
\label{tab:eval3}
\scalebox{0.6}{
\begin{tabular}{ccccccc}
\multicolumn{1}{l}{}                            & \multicolumn{1}{l}{} & \multicolumn{5}{l}{}                                                                    \\ 
\toprule
\multirow{2}{*}{Model}                          & \multicolumn{3}{c}{HitRate@k}                            & \multicolumn{3}{c}{MRR@k}                           \\
                                                & 5                    & 10              & 20              & 5               & 10              & 20              \\ 
\midrule
\begin{tabular}[c]{@{}c@{}}R-GCN\\\end{tabular} & x                    & x               & x               & x               & x               & x               \\
\begin{tabular}[c]{@{}c@{}}\name\\\end{tabular}  & \textbf{34.14x}      & \textbf{22.02x} & \textbf{15.52x} & \textbf{98.05x} & \textbf{86.75x} & \textbf{51.4x}  \\
\bottomrule
\end{tabular}}
\end{table}

Table~\ref{tab:eval32} shows the results pertaining to selection bias. Observe that when the model is trained only the co-purchase edges, it results in a performance dip. This is because the transitive relationships crafted in the test set cannot be captured when trained only on the co-purchase edges. However, R-GCN trained on both co-purchase and co-view edges delivers a poor performance which confirms our claim that we need to provide differential treatments to co-purchase and co-view edges during feature aggregation. \name~delivers the best performance compared to R-GCN model. These results demonstrate how \name~is able to mitigate selection and this answers \textbf{EQ5}.
\begin{table}
\centering
\caption{Selection bias evaluation on G1 graph}
\label{tab:eval32}
\scalebox{0.7}{
\begin{tabular}{ccccccc}
\multicolumn{1}{l}{}                                                 & \multicolumn{1}{l}{} & \multicolumn{5}{l}{}                                                               \\ 
\toprule
\multirow{2}{*}{Model}                                               & \multicolumn{3}{c}{HitRate@k}                         & \multicolumn{3}{c}{MRR@k}                         \\
                                                                     & 5                    & 10            & 20             & 5              & 10             & 20              \\ 
\midrule
\begin{tabular}[c]{@{}c@{}}R-GCN\\(co-purchase+co-view)\end{tabular} & x                    & x             & x              & x              & x              & x               \\
\begin{tabular}[c]{@{}c@{}}\name\\(co-purchase)\end{tabular}          & 3.1x                 & 3.26x         & 3.5x           & 4.26x          & 4.02x          & 3.94x           \\
\begin{tabular}[c]{@{}c@{}}\name\\(co-purchase+co-view)\end{tabular}  & \textbf{3.43x}       & \textbf{3.5x} & \textbf{3.59x} & \textbf{4.58x} & \textbf{4.32x} & \textbf{4.21x}  \\
\bottomrule
\end{tabular}}
\end{table}

In order to answer \textbf{EQ6}, we assume a query product $q$. We generate a related product $v$, when $\theta^s_q.\theta^t_v$ is the highest $\forall v \in P$. Similarly, given a query product $q$, we generate a similar product $v$, when $\theta^s_q.\theta^s_v$ is the highest $\forall v \in P$. Observe from Table~\ref{tab:eval7} the kind of recommendations generated. This shows that we can leverage both the source and target embeddings generated by \name~to recommend not only related products, but also similar items as a byproduct. 
\begin{table}[!htb]
\centering
\caption{Sample product recommendations generated using \name}
\label{tab:eval7}
\scalebox{0.7}{
\begin{tabular}{cccc} 
\toprule
\textbf{Sl. No} & \textbf{query product}                                                                            & \begin{tabular}[c]{@{}c@{}}\textbf{related product}\\\textbf{recommendation}\end{tabular} & \begin{tabular}[c]{@{}c@{}}\textbf{similar product}\\\textbf{recommendation}\end{tabular}             \\ 
\midrule
1               & \begin{tabular}[c]{@{}c@{}}amway attitude \\\uline{sunscreen} cream 100 g\end{tabular}                    & \begin{tabular}[c]{@{}c@{}}amway attitude \uline{face} \\\uline{wash} for dry skin - 100 ml\end{tabular}      & \begin{tabular}[c]{@{}c@{}}amway attitudes \uline{sunscreen} \\cream 100 g x 2 = 200 gm\end{tabular}           \\
2               & \begin{tabular}[c]{@{}c@{}}oppo reno \uline{phone} 6 pro 5g stellar \\black 12gb ram 256gb storage\end{tabular} & \begin{tabular}[c]{@{}c@{}}lustree oppo reno 6 pro 5g \\\uline{bumper back cover case}~\end{tabular}  & \begin{tabular}[c]{@{}c@{}}oppo reno \uline{phone} 6 5g stellar \\black 8gb ram 128gb storage\end{tabular}           \\
3               & \begin{tabular}[c]{@{}c@{}}paseo \uline{tissues} pocket \\hanky 6 packs 3 ply\end{tabular}                & \begin{tabular}[c]{@{}c@{}}origami \uline{wet wipes} wet \\tissue wet facial tissue\end{tabular}      & \begin{tabular}[c]{@{}c@{}}inkulture pocket hanky \uline{tissue} \\paper white pack of 10\end{tabular}        \\
4               & \begin{tabular}[c]{@{}c@{}}yiwoo 5 pieces false \\\uline{eyelashes curler}~\end{tabular}                  & \begin{tabular}[c]{@{}c@{}}ardell duo individual \uline{lash} \\\uline{adhesive} white 7 g\end{tabular}       & vega premium \uline{eye lash curler}                                                                          \\
5               & \begin{tabular}[c]{@{}c@{}}skinn by titan steele \\\uline{fragrance} for men 100 ml\end{tabular}  & \begin{tabular}[c]{@{}c@{}}blue nectar natural vitamin c \uline{face} \\\uline{cream} for glowing skin~\end{tabular} & \begin{tabular}[c]{@{}c@{}}skinn by titan \\\uline{fragrance} for men 20ml\end{tabular} \\
\bottomrule
\end{tabular}}
\end{table}

\subsection{Online Platform Performance}
\label{abtest}
We further evaluate the performance of \name~in production environment by conducting an A/B test in two different marketplaces. For the control group, we use an incumbent approach based on product co-purchases, while for the treatment group, we show the recommendations generated from \name. We run the experiments for $4$ weeks and observe +170\% improvement on product sales and +190\% improvement on profit gain. All the results are statistically significant with p-value $<0.05$. These results show the product recommendations generated from \name~can significantly improve customer shopping experience in discovering potentially related products of interest.

\section{Conclusion and Future Work}
In this paper, we propose \name, a novel Graph Neural Network based framework for related product recommendation, wherein the problem is formulated as a node recommendation task on a directed product graph. In order to train the model, we employ an asymmetric loss function by modelling product co-purchase likelihood, capturing product asymmetry and mitigating selection bias. We leverage the trained GNN model to learn dual embeddings for each product, by appropriately aggregating features from its neighborhood. Extensive offline experiments show that \name~outperforms state-of-the-art baselines for the node recommendation and link prediction tasks against state-of-the-art baselines. 
In the future, we will explore the possibility of representing products using complex valued embeddings and extend \name~to the case of complex embeddings to improve the performance. 
\bibliographystyle{splncs04}
\bibliography{main}

\end{document}